\documentclass[11pt]{amsart}

\usepackage{amsfonts}
\usepackage{amsmath}
\usepackage{amssymb}
\usepackage{graphicx}%
\usepackage {epsf}
\usepackage{verbatim}
 \theoremstyle{plain}

\newtheorem{corollary}{Corollary}[section]

\newtheorem{definition}{Definition}[section]

\newtheorem{lemma}{Lemma}[section]

\newtheorem{proposition}{Proposition}[section]

\newtheorem{theorem}{Theorem}[section]

\numberwithin{equation}{section}

\begin{document}

\vspace{0.5in}

\title[A New Causal Topology and Why the Universe is Co-compact]{A New Causal Topology and Why \\ the Universe is Co-compact}
\author{Martin Maria Kov\'{a}r}
\address{Department of Mathematics, Faculty of Electrical Engineering and
Communication, University of Technology, Technick\'{a} 8, Brno, 616 69, Czech Republic}
\email{kovar@feec.vutbr.cz}
\subjclass[2000]{} 

\begin{abstract} We show that there exists a canonical topology, naturally
connected with the causal site of J. D. Christensen and L. Crane, a pointless algebraic structure
motivated by quantum gravity. Taking a causal site compatible with Minkowski space, on every
compact subset our topology became a~reconstruction of the original topology of the spacetime (only
from its causal structure). From the global point of view, the reconstructed topology is the de
Groot dual or co-compact with respect to the original, Euclidean topology. The result indicates
that the causality is the primary structure of the spacetime, carrying also its topological
information.

\end{abstract}


\maketitle

\renewcommand\theenumi{\roman{enumi}}
\renewcommand\theenumii{\arabic{enumii}}
\font\eurb=eurb9 \font\seurb=eurb7
\def\cl{\operatorname{cl}}
\def\iff{if and only if }
\def\sup{\operatorname{sup}}
\def\clt{\operatorname{cl}_\theta}
\def\cli#1{\operatorname{cl}_{#1}}
\def\inti#1{\operatorname{int}_{#1}}
\def\int{\operatorname{int}}
\def\intt{\operatorname{int}_\theta}
\def\id#1{\operatorname{\text{\sl id}}_{#1}}
\def\ord{\operatorname{ord}}
\def\SIGMA{\operatorname{\Sigma}}
\def\cf{\operatorname{cf}}
\def\diag{\operatorname{\Delta}}
\def\Mezera{\vskip 15 mm}
\def\mezera{\bigskip}
\def\Mezerka{\medskip}
\def\mezerka{\smallskip}
\def\iff{if and only if }
\def\G{\frak G}
\def\A{\Bbb A}
\def\I{\Bbb I}
\def\C{\mathcal C}
\def\F{\mathcal F}
\def\L{\mathcal L}
\def\P{\mathcal P}
\def\B{\frak B}
\def\O{\mathcal O}
\def\T{\frak T}
\def\X{\frak X}
\def\S{\mathcal S}
\def\K{\mathcal K}
\def\Q{\Bbb Q}
\def\R{\Bbb R}
\def\N{\Bbb N}
\def\D{\Bbb D}
\def\Ds{\mathcal D}
\def\T{\mathcal T}
\def\zero{\bold 0}
\def\m{\frak m}
\def\n{\frak n}
\def\ts{  space }
\def\nbd{neighborhood }
\def\nbds{neighborhoods }
\def\card{cardinal }
\def\implies{\Rightarrow }
\def\map{\rightarrow }
\def\ekv{\Leftrightarrow }
\def\gre{\succcurlyeq}
\def\ngre{\not\succcurlyeq}
\def\gr{\succ}
\def\lre{\preccurlyeq}
\def\Ty{$T_{3.5}$ }
\def\Tyk{$T_{3.5}$}
\def\Slc{\text{\eurb Slc}}
\def\Top{\text{\eurb Top}}
\def\top{\text{\seurb Top}}
\def\Comp{\text{\eurb Comp}}
\def\TReg{\text{\eurb $\Theta$-Reg}}
\def\eK{\text{\eurb K}}
\def\M{\Bbb M}
\def\up{\uparrow\!\!}
\def\down{\downarrow\!\!}
\def\meet{\wedge}
\def\join{\vee}
\def\causeq{\preccurlyeq}
\def\caus{\prec}

\section{Introduction}

\medskip

The belief that the causal structure of spacetime is its most fundamental underlying structure is
almost as old as the idea of the relativistic spacetime itself. But how it is related to the
topology of spacetime? By tradition, there are no doubts regarding the topology of spacetime at
least locally, since it is considered to be locally homeomorphic with the cartesian power of the
real line, equipped with the Euclidean topology. But more recently, there appeared concepts of
discrete and pointless models of spacetime in which the causal structure is introduced
axiomatically and so independently on the locally Euclidean models. Is, in these cases, the
axiomatic causal structure rich enough to carry also the full topological information? And, after
all, how the topology that we perceive around us and which is essentially and implicitly at the
background of many physical phenomena, may arise?

In this paper we introduce a general construction, suitable for equipping a set of objects with a
topology-like structure, using the inner, natural and intuitive relationships between them. We use
the construction to show that another algebraic structure, motivated by the research in quantum
geometry and gravitation -- the causal site of J. D. Christensen and L. Crane -- very naturally
generates a compact T$_1$ topology on itself. Testing the construction on Minkowski space we show
that coming out from its causality structure, the universe -- in its first approximation
represented by Minkowski space -- naturally has so called co-compact topology (also called the de
Groot dual topology) which is compact, superconnected, T$_1$ and non-Hausdorff. The co-compact
topology on Minkowski space coincides with the Euclidean topology on all compact sets -- in the
more physically related terminology, at the finite distances. Therefore, the studied construction
has probably no impact to the description of local physical phenomena, but it changes the global
view at the universe. Perhaps it could help to explain how the topology that we perceive ``around
us" (in any way -- by our everyday experience, as well as by experiments, measurements and other
physical phenomena) may arise from causality.

\bigskip

\section{Mathematical Prerequisites}\label{prerequisites}

Throughout this paper, we mostly use the usual terminology of general topology, for which the
reader is referred to \cite{Cs} or \cite{En}, with one exception -- in a consensus with a modern
approach to general topology, we no longer assume the Hausdorff separation axiom as a part of the
definition of compactness. This is especially affected by some recent motivations from computer
science, but also the contents of the paper \cite{HPS} confirms that such a modification of the
definition of compactness is a relevant idea. Thus we say that a topological space is {\it
compact}, if every its open cover has a finite subcover, or equivalently, if every centered system
of closed sets or a closed filter base has a non-empty intersection. Note that by the well-known
Alexander's subbase lemma, the general closed sets may be replaced by more special elements of any
closed subbase for the topology.

We have already  mentioned the co-compact or the de Groot dual topology, which was first
systematically studied probably at the end of the 60's by  J. de Groot and his coworkers, J. M.
Aarts, H. Herrlich, G. E. Strecker and E. Wattel. The initial paper is \cite{Gro}. About 20 years
later the co-compact topology again came to the center of interest of some topologists and
theoretical computer scientists in connection with their research in domain theory. During
discussions in the community the original definition due to de Groot was slightly changed to its
current form, inserting a word ``saturated" to the original definition (a set is saturated, if it
is an intersection of open sets, so in a T$_1$ space, all sets are saturated). Let $(X,\tau)$ be a
topological space. The topology generated by the family of all compact saturated sets used as the
base for the closed sets, we denote by $\tau^G$ and call it {\it co-compact} or {\it de Groot} dual
with respect to the original topology $\tau$. In \cite{LM} J. Lawson and M. Mislove stated
question, whether the sequence, containing the iterated duals of the original topology, is infinite
or the process of taking duals terminates after finitely many steps with topologies that are dual
to each other. In 2001 the author solved the question and proved that only 4 different topologies
may arise (see \cite{Kov1}).

The following theorem summarizes the previously mentioned facts important for understanding the
main results, contained in Section~\ref{causal}. The theorem itself is not new, under slightly
different terminology the reader can essentially find it in \cite{Gro}. A more general result,
equivalently characterizing the topologies satisfying $\tau=\tau^{GG}$, the reader may find in the
second author's paper \cite{Kov3}. For our purposes, the reader may replace a general non-compact,
locally compact Hausdorff space by the Minkowski space equipped with the Euclidean topology. The
proof we present here only for the reader's convenience, without any claims of originality. For the
proof we need to use the following notion. Let $\psi$ be a family of sets. We say that $\psi$ has
the finite intersection property, or briefly, that $\psi$ has {\it f.i.p.}, if for every $P_1,
P_2,\dots, P_k\in\psi$ it follows $P_1\cap P_2\cap \dots\cap P_k\ne\varnothing$. In some literature
(for example, in \cite{Cs}), a collection $\psi$ with this property is called {\it centered}.

\bigskip

\begin{theorem}\label{degroot} Let $(X,\tau)$ be a non-compact, locally compact Hausdorff topological space. Then
\begin{enumerate}
\item $\tau^{G}\subseteq \tau$, \item $\tau=\tau^{GG}$, \item $(X,\tau^{G})$ is compact and
superconnected, \item the topologies induced from $\tau$ and $\tau^G$ coincide on every compact
subset of $(X,\tau)$.
\end{enumerate}
\end{theorem}

\begin{proof} The topology $\tau^G$ has a closed base  which consists of  compact sets.
Since in a Hausdorff space all compact sets are closed, we have (i).

Let $C\subseteq X$ be a closed set with respect to $\tau$, to show that $C$ is compact with respect
to $\tau^G$, let us take a non-empty family $\Phi$ of compact subsets of $(X,\tau)$, such that the
family $\{C\}\cup\Phi$ has f.i.p. Take some $K\in\Phi$. Then the family $\{C\cap K\}\cup\{C\cap
F|\, F\in \Phi\}$ also has f.i.p. in a compact set $K$, so it has a non-empty intersection. Hence,
also the intersection of $\{C\}\cup\Phi$ is non-empty, which means that $C$ is compact with respect
to $\tau^G$. Consequently, $C$ is closed in $(X,\tau^{GG})$, which means that $\tau\subseteq
\tau^{GG}$. The topology $\tau^{GG}$ has a closed base consisting of sets which are compact in
$(X,\tau^{G})$. Take such a set, say $H\subseteq X$. Let $x\in X\smallsetminus H$. Since $(X,\tau)$
is locally compact and Hausdorff, for every $y\in H$ there exist $U_y, V_y\in \tau$ such that $x\in
U_y$, $y\in V_y$ and $U\cap V=\varnothing$, with $\cl U_y$ compact. Denote $W_y=X\smallsetminus\cl
U_y$. We have $y\in V_y\subseteq W_y$, so the sets $W_y$, $y\in H$ cover $H$. The complement of
$W_y$ is compact with respect to $\tau$, so $W_y\in \tau^G$. The family $\{W_y|\,y\in H\}$ is an
open cover of the compact set $H$ in $(X,\tau^G)$, so it has a finite subcover, say
$\{W_{y_1},W_{y_2},\dots,W_{y_k}\}$. Denote $U=\bigcap_{i=1}^k U_{x_i}$. Then $U\cap
H=\varnothing$, $x\in U\subseteq X\smallsetminus H$, which means that $X\smallsetminus H\in\tau$
and $H$ is closed in $(X,\tau)$. Hence, $\tau^{GG}\subseteq \tau$, an together with the previously
proved converse inclusion, it gives (ii).

Let us show (iii). Take any collection $\Psi$ of compact subsets of $(X,\tau)$ having f.i.p. They
are both compact and closed in $(X,\tau)$, so $\bigcap\Psi\ne\varnothing$. Then $(X,\tau^G)$ is
compact. Let $U,V\in\tau^G$ and suppose that $U\cap V=\varnothing$. The complements of $U$, $V$ are
compact in $(X,\tau)$ as intersections of compact closed sets in a Hausdorff space. Then $(X,\tau)$
is compact as a union of two compact sets, which is not possible. Hence, it holds (iii).

Finally, take a compact subset $K$ and a closed subset $C$ of $(X,\tau)$. Then $K\cap C$ is compact
in $(X,\tau)$ and hence closed  in $(X,\tau^G)$. Thus the topology on $K$ induced from $\tau^G$ is
finer than the topology induced from $\tau$. Together with (i), we get (iv).
\end{proof}

\medskip

\bigskip

\section{How to Topologize Everything}

\medskip

As it has been recently noted in \cite{HPS}, the nature or  the physical universe, whatever it is,
has probably no existing, real points like in the classical Euclidean geometry (or, at least, we
cannot be absolutely sure of that). Points, as a useful mathematical abstraction, are
infinitesimally small and thus cannot be measured or detected by any physical way. But what we can
be sure that really exists, there are various locations, containing concrete physical objects. In
this paper we will call these locations {\it places}. Various places can overlap, they can be
merged, embedded or glued together, so the theoretically understood virtual ``observer" can visit
multiple places simultaneously. For instance, the Galaxy, the Solar system, the Earth, (the
territory of) Europe, Brno (a beautiful city in Czech Republic, the place of author's residence),
the room in which the reader is present just now, are simple and natural examples of places
conceived in our sense. Certainly, in this sense,  one can be present at many of these places at
the same time, and, also certainly, there exist pairs of places, where the simultaneous presence of
any physical objects is not possible. Or, at least, from our everyday experience it seems the
nature behaves in this way. Thus the presence of various physical objects connects these primarily
free objects -- our places -- to the certain structure, which we call a {\it framework}.

Note that it does not matter that the places are, at the first sight, determined rather vaguely or
with some uncertainty. They are conceived as elements of some algebraic structure, with no any
additional geometrical or metric structure and as we will see later, the ``uncertainty" could be
partially eliminated by the relationships between them. Let's now give the precise definition.

\begin{definition} Let $\P$ be a set, $\pi\subseteq 2^\P$. We say that $(\P,\pi)$ is a
framework. The elements of $\P$ we call {\it places}, the set $\pi$ we call {\it framology}.
\end{definition}

Although every topological space is a framework by the definition, the elementary interpretation of
a framework is very different from the usual interpretation of a topological space. The elements of
the framology are not primarily considered as neighborhoods of places, although this seems to be
also very natural. If $\P$ contains all the places that are or can be observed, the framology $\pi$
contains the list of observations of the fact that the virtual ``observer" or some physical object
that ``really exists" (whatever it means), can be present at some places simultaneously. The
structure which $(\P,\pi)$ represents arises from these observations.

\medskip

Let us introduce some other useful notions.

\begin{definition} Let $(P,\pi)$ and $(S,\sigma)$ be frameworks. A mapping $f:P\map S$ satisfying
$f(\pi)\subseteq\sigma$ we call a {\it framework morphism}.
\end{definition}

\begin{definition} Let $(P,\pi)$ be a framework, $\sim$ an equivalence relation on $P$. Let
$P_\sim$ be the set of all equivalence classes and $g:P\map P_\sim$ the corresponding quotient map.
Then $(P_\sim, g(\pi))$ is called the quotient framework of $(P,\pi)$ (with respect to the
equivalence $\sim$).
\end{definition}

\begin{definition}A framework $(P,\pi)$ is T$_0$ if for every $x,y\in P$, $x\ne y$, there
exists $U\in\pi$ such that $x\in U$, $y\notin U$ or $x\notin U$, $y\in U$.
\end{definition}

\begin{definition} Let $(P,\pi)$ be a framework. Denote $P^d=\pi$ and $\pi^d=\{\pi(x)|\, x\in P\}$,
where $\pi(x)=\{U|\, U\in\pi, x\in U\}$. Then $(P^d, \pi^d)$ is the {\it dual} framework of
$(P,\pi)$. The places of the dual framework $(P^d, \pi^d)$ we call {\it abstract points} or simply
{\it points} of the original framework $(P,\pi)$.
\end{definition}

The framework duality is a simple but handy tool for switching between the classical point-set
representation (like in topological spaces) and the point-less representation, introduced above.

\bigskip

{\bf Some Examples.} There is a number of  natural examples of mathematical structures satisfying
the definition of a framework, including non-oriented graphs, topological spaces (with open maps as
morphisms), measurable spaces or texture spaces  of M. Diker \cite{Di}. Among physically motivated
examples, we may mention the Feynman diagrams with particles in the role of places and interactions
as the associated abstract points. Very likely, certain aspects of the string theory, related to
general topology, can also be formulated in terms of the framework theory.

\bigskip

It should be noted that the notion of a framework is a special case of the notion of the {\it
formal context}, due to B. Ganter and R. Wille \cite{GW}, sometimes also referred as the Chu space
\cite{ChL}. Recall that a formal context is a triple $(G,M, I)$, where $G$ is a set of objects, $M$
is a set of attributes and $I\subseteq G\times M$ is a binary relation. Thus a framework $(P,\pi)$
may be represented as a formal context $(P,\pi, \in)$, where objects are the places and their
attributes are the abstract points.  Even though the theory and methods of formal concept analysis
may be a useful tool also for our purposes, we prefer the topology-related terminology that we
introduced in this section because it seems to be more close to the way, how mathematical physics
understands to the notion of spacetime. It also seems that frameworks are closely related to the
notion of partial metric due to S. Matthews \cite{Ma}, but these relationships will be studied in a
separate paper.

\begin{proposition}Let $(P,\pi)$ be a framework. Then $(P^d,\pi^d)$ is T$_0$.
\end{proposition}

\begin{proof} Denote $S=\pi$, $\sigma=\{\pi(x) |\, x\in P\}$, so $(S,\sigma)$ is the dual framework
of $(P, \pi)$. Let $u, v\in S$, $u\ne v$. Since $u, v\in 2^P$ are different sets, either there
exists $x\in u$ such that $x\notin v$, or there exists $x\in v$, such that $x\notin u$. Then
$u\in\pi(x)$ and $v\notin\pi(x)$, or $v\in\pi(x)$ and $u\notin\pi(x)$. In both cases there exists
$\pi(x)\in\sigma$, containing one element of $\{u, v\}$ and not containing the other.
\end{proof}

\begin{theorem} Let $(P,\pi)$ be a framework. Then $(P^{dd}, \pi^{dd})$ is isomorphic to the
quotient of $(P,\pi)$. Moreover, if $(P,\pi)$ is T$_0$, then $(P^{dd}, \pi^{dd})$ and $(P,\pi)$ are
isomorphic.
\end{theorem}

\begin{proof} We denote $R=P^d=\pi$, $\rho=\pi^d=\{\pi(x) |\, x\in P \}$,
$S=R^d=\rho$, $\sigma=\rho^d=\{\rho(x) |\, x\in R \}$. Then $(S, \sigma)$ is the double dual of
$(P,\pi)$. It remains to show, that $(S, \sigma)$ is isomorphic to some quotient of $(P,\pi)$.

For every $x\in P$, we put $f(x)=\pi(x)$. Then $f:P\map S$ is a surjective mapping. It is easy to
show, that $f$ is a morphism. Indeed, if $U\in\pi$, then $f(U)=\{\pi(x) |\, x\in U\}=\{\pi(x) |\,
x\in P,  U\in\pi(x) \}=\{V |\, V\in\rho,  U\in V\}=\rho(U)\in\sigma$. Therefore,
$f(\pi)\subseteq\sigma$, which means that $f$ is an epimorphism of the framework $(P,\pi)$ onto
$(S,\sigma)$.

Now, we define $x\sim y$ for every $x, y\in P$ \iff $f(x)=f(y)$. Then $\sim $ is an equivalence
relation on $P$. For every equivalence class $[x]\in P_\sim$ we put $h([x])=f(x)$. The mapping $h:
P_\sim\map S$ is correctly defined, moreover, it is a bijection. The verification that $h$ is a
framework isomorphism is standard, but, because of completeness, it has its natural place here. The
quotient framology on $P_\sim$ is $g(\pi)$, where $g:P\map P_\sim$ is the quotient map. The
quotient map $g$ satisfies the condition $h\circ g=f$. Let $W\in g(\pi)$. There exists $U\in \pi$
such that $W=g(U)$. Then $h(W)=h(g(U))=f(U)\in\sigma$. Hence $h(g(\pi))\subseteq\sigma$, which
means that $h: P_\sim\map S$ is a framework morphism. Conversely, let $W\in\sigma=\{\rho(U) |\,
U\in\pi\}$. We will show that $h^{-1}(W)\in g(\pi)$. By the previous paragraph, $\rho(U)=f(U)$ for
every $U\in\pi$, so there exists $U\in\pi$, such that $W=f(U)=h(g(U))$. Since $h$ is a bijection,
it follows that $h^{-1}(W)=g(U)\in g(\pi)$. Hence, also $h^{-1}:S \map P_\sim$ is a framework
morphism, so the frameworks $(P_\sim, g(\pi) )$ and $(S,\sigma)$ are isomorphic.

Now let us consider the special case when $(P,\pi)$ is T$_0$. Suppose that $f(x)=f(y)$ for some
$x,y\in P$. Then $\pi(x)=\pi(y)$, which is possible only when $x=y$. Then the relation $\sim$ is
the diagonal relation, and the quotient mapping $g$ is an isomorphism.
\end{proof}

\begin{corollary} Every framework arise as dual \iff it is T$_0$.
\end{corollary}

\begin{corollary} For every framework $(P,\pi)$, it holds  $(P^d,\pi^d)\cong (P^{ddd},
\pi^{ddd})$.
\end{corollary}

\bigskip

\section{Topology of Causal Sites}\label{causal}

\medskip

In this section we show that the notion of a~framework, introduced and studied in the previous
section, has some real utility and sense. In a contrast to simple examples mentioned above, from a
properly defined  framework we will be able to construct a topological structure with a real
physical meaning.

\medskip

Recall that a {\it causal site} $(S,\sqsubseteq, \prec)$ defined by J. D. Christensen and L. Crane
in \cite{CC} is a set $S$ of {\it regions} equipped with two binary relations $\sqsubseteq$,
$\prec$, where $(S,\sqsubseteq)$ is a partial order having the binary suprema $\sqcup$ and the
least element $\bot\in S$, and $(S\smallsetminus\{\bot\},\prec)$ is a strict partial order (i.e.
anti-reflexive and transitive), linked together by the following axioms, which are satisfied for
all regions $a, b, c\in S$:
\begin{enumerate}

\item $b\sqsubseteq a$ and $a\prec c$ implies $b\prec c$,

\item $b\sqsubseteq a$ and $c\prec a$ implies $c\prec b$,

\item $a\prec c$ and $b\prec c$ implies $a\sqcup b\prec c$.

\item There exits $b_a\in S$, called {\it cutting of $a$ by $b$}, such that {
     \begin{enumerate}

     \item $b_a\prec a$ and $b_a\sqsubseteq b$;

     \item if $c\in S$, $c\prec a$ and $c\sqsubseteq b$ then $c\sqsubseteq b_a$.

     \end{enumerate}

}
\end{enumerate}

\medskip

Consider a causal site $(P,\sqsubseteq, \prec)$ and let us define appropriate framework structure
on $P$. We say that a subset $F\subseteq P$ set is centered, if for every $x_1, x_2, \dots, x_k\in
F$ there exists $y\in P$, $y\ne\bot$ satisfying $y\sqsubseteq x_i$ for every $i=1,2,\dots, k$. If
$\L\subseteq 2^P$ is a chain of centered subsets of $P$ linearly ordered by the set inclusion
$\subseteq$, then $\bigcup \L$ is also a centered set. Then every centered $F\subseteq P$ is
contained in some maximal centered $M\subseteq P$. Let $\pi$ be the family of all maximal centered
subsets of $P$. Now, consider the framework $(P,\pi)$ and its dual $(P^d, \pi^d)$. Let $(X,\tau)$
be the topological space with $X=P^d=\pi$ and the topology $\tau$ generated by its closed subbase
(that is, a subbase for the closed sets) $\pi^d$.

\begin{theorem}\label{comp} The topological space $(X,\tau)$, corresponding to the framework $(P^d,\pi^d)$ and the causal
site $(P,\sqsubseteq, \prec)$, is compact T$_1$.
\end{theorem}

\begin{proof} By the well-known Alexander's subbase lemma, for proving the compactness of $(X,\tau)$ it is sufficient
to show, that any subfamily of $\pi^d$ having the f.i.p., has nonempty intersection. The subbase
for the closed sets of $(X,\tau)$ has the form $\pi^d=\{\pi(x)|\, x\in P\}$, so any subfamily of
$\pi^d$ can be indexed by a subset of $P$. Let $F\subseteq P$ and suppose that for every $x_1,
x_2,\dots, x_k\in F$ we have
$$\pi(x_1)\cap\pi(x_2)\cap\dots\cap\pi(x_k)\ne\varnothing.$$ Then there exists $U\in\pi$ such that
$U\in\pi(x_1)\cap\pi(x_2)\cap\dots\cap\pi(x_k)$, so $x_i\in U$ for every $i=1,2,\dots, k$. Since
$U$ is a (maximal) centered family, there exists $\bot\ne y\in P$ such that $y\sqsubseteq x_i$ for
every $i=1,2,\dots,k$. Thus $F$ is a centered family, contained in some maximal centered family
$M\subseteq P$. But then we have $M\in\pi$, so $$M\in\bigcap_{x\in M}\pi(x)\subseteq\bigcap_{x\in
F}\pi(x)\ne\varnothing.$$ Hence, $(X,\tau)$ is compact.

Let $U,V\in X=\pi$, $U\ne V$. Since both are maximal centered subfamilies of $P$, none of them can
contain the other one. So, there exist $x, y\in P$ such that $x\in U\smallsetminus V$ and $y\in
V\smallsetminus U$. Then $U\in\pi(x)$, $V\notin\pi(x)$, $V\in\pi(y)$, $U\notin\pi(y)$. Thus
$X\smallsetminus\pi(x)$, $X\smallsetminus\pi(y)$ are open sets in $(X,\tau)$ containing just one of
the points $U, V$. So the topological space $(X,\tau)$ satisfies the T$_1$ axiom.
\end{proof}

\bigskip

The motivation for introducing and studying the notion of a causal site lies especially in the hope
that it may be helpful in formulation and solution of certain problems in quantum gravity.
Especially in those situations, in which the traditional models are less convenient or even may
fail (see \cite{CC} for more detail). In this situations, possibly very different from our
macroscopic, everyday experience, also the topological structure of spacetime is an important and
legitimate subject of the research. This is one of the possible motivations for the topology that
we have introduced by the way described above and also a good motivation for Theorem~\ref{comp}.
Another, perhaps even more important motivation it is to investigate how the topology of spacetime,
which is perceived in the reality and implicitly is involved in physical phenomena, arises. So the
first question we should ask it is whether the corresponding topology, constructed from the causal
site by the described way has any physical meaning. But how to do that? Certainly, first we must
test the construction at those situations that are working and well understood in the scope of the
classical, traditional models. That is why we choose Minkowski space and its causal structure for
the next considerations. If our previous construction is worth, then the output topology that we
receive should be closely related to the Euclidean topology on $\M$.

\bigskip

In \cite{CC}, the authors show that the definition of a causal site is compatible with the inner
structure of the Minkowski space. Moreover, it is also shown that the same is true for the stably
causal Lorentzian manifold (for the precise definition of stable causality see \cite{CC}; by a
result of  S. Hawking and  G. Ellis \cite{HE},  it is equivalent to the existence of a global time
function).  However, it is easy to check that the causal site compatible with the stably causal
Lorentzian manifold need be not unique. As we will see later, for the purposes of reconstruction of
the topology from the causal structure we need much finer setting for the corresponding causal
site, than it is used in the two simple examples of the paper \cite{CC}.

\bigskip

Let us denote by $\M=\R^4$ the Minkowski space. Recall  that it has a natural structure of a real,
$4$-dimensional vector space, equipped with the bilinear form $\eta:\M\times\M\map\R$, called the
Minkowski inner product. The Minkowski modification of the inner product is not positively definite
as the usual inner product, but in the standard basis it is represented by the diagonal matrix with
the diagonal entries $(1, -1, -1, -1)$. Then a vector $v\in\M$ is called timelike, if
$\eta(v,v)>0$, lightlike or null if $\eta(v,v)=0$ and spacelike, if $\eta(v,v)<0$. Further, the
vector $v$ is said to be future-oriented, if its first coordinate, which represents the time, is
positive. Similarly, $v$ is past-oriented, if its first coordinate is negative. We write $v\ll w$
for $v,w \in \M$ if the vector $w-v$ is timelike and future-oriented. In \cite{CC} the sets of the
form $D(p, q)=\{x|\, x\in\M, p\ll x\ll q\}$ are called diamonds. They are used for the construction
of an example of a certain causal site. In this setting, diamonds are open sets in the Euclidean
topology, bounded by two light cones at points $p, q\in \M$. It is not difficult to show that open
diamonds form a base for the Euclidean topology on $\M$. However, for the purpose of a
reconstruction of the topology from the causal structure it is more convenient to consider the
closed variant of diamonds (with respect to the Euclidean topology).

\medskip

We define $p\leqslant q$ if the vector $q-p$ is non-past-oriented and non-spacelike, that is, if
its time coordinate is non-negative and $\eta(q-p,q-p)\ge 0$.  We also denote $\zero=(0,0,0,0)$.
Now, we put
$$J^+(p)=\{x|\, x\in\M, p\leqslant x\},$$
$$J^-(p)=\{x|\, x\in\M, x\leqslant p\}$$
and
$$J(p)=J^+(p)\cup J^-(p).$$

\bigskip

Let $\Vert\!\cdot\!\Vert$ be the Euclidean norm on $\M$. For a real number  $\varepsilon>0$ and a
point  $x\in M$, by $B_\varepsilon(x)$ we denote the open ball $B_\varepsilon(x)=\{y|\, y\in \M,
\Vert\!x-y\!\Vert<\varepsilon\}$. The Euclidean topology on $M$, generated by the norm
$\Vert\!\cdot\!\Vert$ and these open balls, we denote by $\tau_E$. The de Groot dual or co-compact
topology on $\M$ we denote by $\tau_E^G$.

\bigskip

For our next considerations we will need several lemmas, which will point out some important
properties of the relation $\leqslant$ and of the cones $J(p)$ in $\M$. We do not claim originality
for these results, only the context in which we will use them -- the construction of a certain
causal site on $\M$ -- is new. Although the results can be essentially found in the literature, in
order to avoid problems with different notation and also for the reader's convenience, we present
here the complete proofs. However, for a more advanced foundations of the conus theory, the reader
is referred to the comprehensive paper \cite{KR}.

\medskip

\begin{lemma}\label{uzk}
 The sets $J^+(\zero)$ and $J^-(\zero)$ are closed with respect to the operation $+$ of the vector space
 $(\M, +)$.
\end{lemma}

\begin{proof}
 Let $x,y\in J^+(\zero)$. Let $x=s+t$, $y=r+u$, where $r,s,t,u\in\M$ and the vectors $r$, $s$ have
zero time coordinate, and the vectors $u$, $t$ have zero space coordinates. Since $x\in
J^+(\zero)$, we have $\eta(x,x)\ge 0$, which is equivalent to $\Vert t\Vert\ge \Vert s\Vert$.
Similarly, from $y\in J^+(\zero)$ we get $\Vert u\Vert\ge\Vert r\Vert$. Since the time coordinates
of $x$, $y$ and so $t$, $u$  are of the same sign, and only one coordinate of $t$, $u$ can be
non-zero, it follows that $\Vert t+u\Vert=\Vert t\Vert+\Vert u\Vert\ge\Vert s\Vert+\Vert
r\Vert\ge\Vert s+r\Vert$. Then $\eta(x+y, x+y)\ge 0$. Since the time coordinate of $x+y$ is
non-negative (as the sum of the non-negative coordinates of $x$, $y$), we finaly get $x+y\in
J^+(\zero)$. The proof for $J^-(\zero)$ is analogous.
\end{proof}

\begin{lemma}
 The binary relation $\leqslant$ is a partial order on $\M$.
\end{lemma}

\begin{proof}
Certainly, $\leqslant$ is reflexive. Suppose that $p\leqslant q$ and $q\leqslant r$ for some
$p,q,r\in\M$. Then $r-p=(r-q)+(q-p)$, so if the time coordinates of $q-p$ and $r-q$ are
non-negative, the same holds also for $r-p$.  Since $\eta(q-p, q-p)\ge 0$ and $\eta(r-q, r-q)\ge
0$, we have $q-p\in J^+(0)$ and $r-q\in J^+(0)$. By Lemma \ref{uzk}, $r-p\in J^+(0)$. Then
$\zero\leqslant r-p$, which gives $\eta(r-p,r-p)\ge 0$. Thus $\leqslant$ is also a transitive
relation.
\end{proof}

\medskip

We denote
$$\Diamond(p,q)=J^+(p)\cap J^-(q),$$
where $p,q\in\M$, $p\leqslant q$. Now let us construct a causal site which reflects causality and
topological properties of  Minkowski space $\M$. Denote $\mathcal D=\{\Diamond(p,q)|\, p,q\in\Q^4,
p\leqslant q\}$. Now, let $(P,\cup, \cap)$ be the set lattice generated by the elements of
$\mathcal D$. Since $P$ can be represented by lattice polynomials (see, e.g. \cite{Gr}), every
element of $P$ can be expressed by unions and intersections of finitely many elements of $\mathcal
D$, it is compact and closed with respect to the Euclidean topology $\tau_E$ on $\M$.

\begin{lemma}
 The family $P$ is a closed base for the co-compact topology on $\M$.
\end{lemma}

\begin{proof}
The co-compact topology $\tau_E^G$ on $\M$ is generated by its open base, which is formed by the
complements of sets, compact in the Euclidean topology $\tau_E$.

Let $K\subseteq\M$ be compact. Denote $U=\M\smallsetminus K$.  Take a point $x\in U$. For every
$y\in K$ there exist $p_y, q_y\in \Q^4$, $p_y\leqslant q_y$, such that $y\in\int \Diamond(p_y,
q_y)$, where the interior is considered with respect to the Euclidean topology $\tau_E$ on $\M$,
and $x\notin \Diamond(p_y, q_y)$. Since $K$ is compact, there exist $y_1, y_2, \dots, y_k\in K$
with
$$K\subseteq \bigcup_{i=1}^k \int \Diamond(p_{y_i}, q_{y_i}).$$ Then
$$x\in \bigcap_{i=1}^k (\M\smallsetminus\Diamond(p_{y_i}, q_{y_i}))=\M\smallsetminus\
\bigcup_{i=1}^k \Diamond(p_{y_i}, q_{y_i})\subseteq U,$$ and the closed set $\bigcup_{i=1}^k
\Diamond(p_{y_i}, q_{y_i})$ is an element of $P$. Hence, also every set $U$, which is open with
respect to  $\tau_E^G$, is a union of complements of elements of $P$, which are closed in the same
topology. Then $P$ forms a closed base for $\tau_E^G$.
\end{proof}

\medskip

Finally, we are ready to complete the construction of the causal site on $\M$. Let $A, B\in P$
non-empty. We put $A\prec B$ if $A\ne B$ and for every $a\in A$, $b\in B$, $a\leqslant b$.

\begin{theorem} $(P,\subseteq,\prec)$ is a causal site.
\end{theorem}

\begin{proof}
First of all, we need to show that $\prec$ is a transitive on the set $P\smallsetminus
\{\varnothing\}$ (the anti-reflexivity of $\prec$ follows directly from the definition). Suppose
that $A\prec B$ and $B\prec C$, where $A, B,C $ are non-empty. Let $a\in A$, $c\in C$. Since
$B\ne\varnothing$, there is some $b\in B$. The vectors $b-a$ and $c-b$ are non-spacelike and
non-past-oriented. Then also the vector $c-a=(c-b)+(b-a)$ is also non-space-like and
non-past-oriented. Suppose that $A=C$. Then $A\prec B$ and $B\prec A$. Taking any $a^\prime\in A$
and $b^\prime\in B$, we get that both vectors $a^\prime-b^\prime$ and $b^\prime-a^\prime$ are
non-spacelike and non-past-oriented, which gives $a^\prime=b^\prime$. Then $A=B$ is a singleton,
but this equality contradicts to the definition of the relation $\prec$. Thus $\prec$ is
transitive.

Since $\subseteq$ is the set inclusion, the axioms (i)-(iii) are satisfied trivially. Let us check
the axiom (iv). Let $A\in P$, $A\ne \varnothing$. Since in the Euclidean topological structure the
compact sets are bounded, there exists a diamond $D=\Diamond(p_0,q_0)$ with $A\subseteq D$. Denote
$$O_A=\{p|\, p\in D, A\subseteq J^+(p)\}.$$
Since $q_0\in O_A$, $O_A\ne\varnothing$. Let $L\subseteq O_A$ be a non-empty linearly ordered chain
with respect to $\leqslant$. We will show that $L$ has an upper bound in $O_A$. Consider the net
$\id L(L,\leqslant)$. Since $D$ is compact, $\id L(L,\leqslant)$ has a cluster point, say $p_L\in
D$. Suppose that there is some $l\in L$ such that $p_L\notin J^+(l)$. Since the set $J^+(l)$ is
closed in $\M$, there exists $\varepsilon >0$ such that $B_\varepsilon(p_L)\cap
J^+(l)=\varnothing$. By the definition of the cluster point, there exists $m\in L$, $l\leqslant m$,
such that $m\in B_\varepsilon(p_L)$. Then $m\in J^+(m)\cap B_\varepsilon(p_L)$, but this is not
possible since $J^+(m)\subseteq J^+(l)$. Hence, $p_L\in \bigcap_{l\in L} J^+(l)$, which means that
$p_L$ is an upper bound of $L$ in $D$. It remains to show that $A\subseteq J^+(p_L)$. Suppose
conversely, that there exists some $r\in A\smallsetminus J^+(p_L)$. Since $J^+(p_L)$ is closed in
$\M$, there exists $\varepsilon >0$ such that $B_\varepsilon(r)\cap J^+(p_L)=\varnothing$. Since
$p_L$ is a cluster point of the net $\id L(L,\leqslant)$, there exists $n\in L$, $n\in
B_{\varepsilon / 2}(p_L)$. Then $r\in A\subseteq J^+(n)$. Denote $q=r+(p_L-n)$. The vector $q$ is
the translation of $r$ by the vector $p_L-n$, and $J^+(p_L)$ is the translation of the cone
$J^+(n)$ by the same vector, so $q\in J^+(p_L)$. Now, $0<\varepsilon\le\Vert r-q\Vert=\Vert
n-p_L\Vert<{\varepsilon\over 2}$, which is a contradiction. Thus $A\subseteq J^+(p_L)$, and so
$p_L\in O_A$ is the upper bound of the chain $L$. Let $M_A$ be the set of all maximal elements of
$O_A$ (with respect to the order $\leqslant $). By Zorn's Lemma, for every $p\in O_A$ there exists
$m\in M_A$ such that $p\leqslant m$. We put
$$A_\bot=\bigcup_{m\in M_A} J^-(m),$$ and for $B\in P$, $B\ne A$ we denote $$B_A=B\cap A_\bot.$$
To claim that $B_A\in P$, we need to show that $M_A$ is finite. The boundary of $A\in P$ can be
decomposed into a finite set $S_A$ of pieces of the boundaries of the cones $J(t)$, $t\in T_A$,
where $T_A$ is a proper finite set. If $m\in M_A$, then the boundary of $J(m)$ must intersect some
elements of $S$, otherwise $m$ cannot be maximal. Moreover, the cone $J(m)$ is fully determined by
a finite and limited number of such intersections, because the points of these intersections must
satisfy the equation of the boundary of $J(m)$. But this would not be possible for an infinite set
$M_A$. Then $B_A\in P$. Let $b\in B_A$, $a\in A$. By the definition of $B_A$, there exists some
$m\in M_A$ with $b\in J^-(m)$, so $b\leqslant m$. We also have $a\in A\subseteq J^+(m)$, so
$m\leqslant a$. Then $b\leqslant a$, which implies $B_A\prec A$.

Suppose that $C\prec A$, $C\subseteq B$ for some $C\in P$. Let $c\in C$. If $a\in A$, then $c\leqslant a$,
which gives $a\in J^+(c)$. Therefore, $A\subseteq J^+(c)$. Then $c\in O_A$, so there exists $m\in M_A$,
such that $c\leqslant m$. Then $c\in J^-(m)\subseteq A_\bot$. Hence, $C\subseteq A_\bot$,
which together with $C\subseteq B$ gives the requested inclusion $C\subseteq B_A$.
\end{proof}

\medskip

Now we will concentrate us on the reconstruction of the original topology on $\M$ from the
causality structure of $(P,\subseteq,\prec)$. Let $\pi$ be the family of all maximal centered
subsets of $P$.

\begin{theorem}
 The topological space $(X,\tau)$ corresponding to the framework $(P^d,\pi^d)$ is homeomorphic to $\M$ equipped with the co-compact topology.
\end{theorem}

\begin{proof}
 As we already defined before, $X=P^d=\pi$. Note that any point $p\in\M$ defines a maximal centred subset of $P$, say $f(p)=\{C|\,
C\in P, p\in C\}$. The family $f(p)$ obviously is centered, since $P$ is closed under finite
intersections and $f(p)$ contains those elements of $P$, whose  contain $p$. Let $Q$ be another
centered family such that $f(p)\subseteq Q\subseteq P$. Suppose that there is some $F\in Q$, such
that $p\notin F$. The set $\M\smallsetminus F$ is open with respect to the Euclidean topology
$\tau_E$, so there exist $u, v\in \Q^4$, $u\leqslant v$, such that $p\in\Diamond(u,v)\subseteq
\M\smallsetminus F$. But $\Diamond(u,v)\in P$, so $\Diamond(u,v)\in f(p)\subseteq Q$, while
$\Diamond(u,v)\cap F=\varnothing$. This contradicts to the assumption that $Q$ is centered. Thus
all elements of $Q$ contain $p$, which means that $Q=f(p)$. Now it is clear that $f(p)$ is a
maximal centred subfamily of $P$.

Conversely, a maximal centered subfamily $Q\in\pi$ has a nonempty intersection, because of
compactness of $\M$ in the co-compact topology. If $\{x,y\}\subseteq \bigcap_{F\in Q}F$,  where
$x\ne y$, then there exist $u, v\in \Q^4$, $u\leqslant v$, such that  $x\in\Diamond(u,v)$ and
$y\notin\Diamond(u,v)$. Then $Q\cup\{\Diamond(u,v)\}\subseteq P$ is an extension of $Q$ which is
also centered, which contradicts to the maximality of $Q$. Thus the intersection of $\bigcap_{F\in
Q}F$ contains only one element, say $g(Q)$. Consequently we have $g(f(p))=p$ and $f(g(Q))=Q$. Thus
the mappings $f:\M\map X$ and  $g:X\map\M$ are bijections inverse to each other. Further, for $A\in
P$ we have $g^{-1}(A)=\{Q|\, Q\in\pi, g(Q)\in A\}=\{Q|\, Q\in\pi, Q\in f(A)\}=\{Q  |\, Q\in\pi,
A\in Q \}=\pi(A)$, which is a subbasic closed set in $(X,\tau)$. Then $g:X\map\M$ is continuous.

Now, take a set $\pi(B)$, where $B\in P$, from the closed base $\pi^d$ of $\tau$. Then
$f^{-1}(\pi(B))=\{p|\, p\in\M, f(p)\in\pi(B)\}=\{p|\, p\in\M, B\in f(p)\}$. For every $p\in
f^{-1}(\pi(B))$, $f(p)$ is a maximal centered subfamily of $P$, containing the set $B$ (which is
compact with respect to $\tau_E$). As we have shown above, its intersection contains the only
element $g(f(p))=p$. So $f^{-1}(\pi(B))=\{p|\, p\in\M, p\in B\}=B$. Since $B$ is a compact set with
respect to the Euclidean topology $\tau_E$ on $\M$, it is closed in the co-compact topology and the
map $f:\M\map X$ is continuous. Hence, the spaces $(X,\tau)$ and $\M$, equipped with the co-compact
topology, are homeomorphic.
\end{proof}

\bigskip

\section{Final Remarks in Historical Context}

\medskip

The progress in mathematical and theoretical physics witnesses that various applications of
topology in physics may be far-reaching and illuminating. It could be very difficult to track down
the origins of such applications, but one of the first attempts may be associated with the year
1914, when A. A. Robb came with his axiomatic system for  Minkowski space $\M$, analogous to the
well-known axioms of Euclidean plane geometry. In \cite{Rb} he essentially proved that the
geometrical and topological structure of $\M$ can be reconstructed from the underlying set and a
certain order relation among its points. As it is noted in \cite{Do},  some prominent
mathematicians and physicists criticized the use of locally Euclidean  topology in mathematical
models of the spacetime. Perhaps as a reflection of these discussions, approximately at the same
time when de Groot wrote his papers on co-compactness duality, there appeared two interesting
papers \cite{Ze} and \cite{Ze2}, in which E. C. Zeeman studied an alternative topology for
Minkowski space. The Zeeman topology, also referred as the fine topology, is the finest topology on
$\M$, which induces the $3$-dimensional Euclidean topology on every space-axis and the
1-dimensional Euclidean topology on the time-axis. Among other interesting properties, it induces
the discrete topology on every light ray. A. Kartsaklis in \cite{Ka} studied connections between
topology and causality. He attempted to axiomatize causality relationships on a point set, equipped
with three binary relations, satisfying certain axioms, by a structure called a {\it causal space}.
He also introduced so called {\it chronological topology}, the coarsest topology, in which every
non-empty intersection of the chronological future and the chronological past of two distinct
points of a causal space is open.

\medskip

In the camp of quantum gravity, there appeared similar efforts and attempts to get some gain from
studying the underlying  structure of spacetime -- topological, geometrical or discrete -- however,
significantly later. The possible motivation is explained, for instance, in \cite{Ro}. C.~Rovelli
notes here that the loop quantum gravity leads to a view of the geometry structure of spacetime at
the short-scale level extremely different from that of a smooth geometry background. Also the
topology of spacetime at Planck scales could be very different from that we meet in our everyday
experience and which has been originally extrapolated from the fundamental concepts of the
continuous and smooth mathematics. Thus the usual properties and attributes of the spacetime, like
its Hausdorffness or metrizability may not be satisfied (for a groundbreaking paper, see
\cite{HPS}). The most important source of inspiration for our paper was the work \cite{CC} of  J.
D. Christensen and L. Crane. Motivated by certain requirements of their research in quantum
gravity, these authors developed a novel axiomatic system for the generalized spacetime, called
{\it causal site}, qualitatively different from the previous, similar attempts. The notion itself
is a successful synthesis of two other notions, a  Grothendieck site (which basically is a small
category equipped with the Groethen\-dieck topology) \cite{Ar} and a  causal set of R. Sorkin
\cite{So}. One of the most important merits of the new axiomatic system it is the fact that the
causal site is a pointless structure, not unlike to some well-known concepts of pointless topology
and locale theory.

\medskip

The contents of our paper can be considered as a certain kind of a virtual experiment. We
constructed a topology from a general causal site by a purely mathematical, straightforward and
canonical way. Taking the causal site given by Minkowski space we did not receive the usual and
naturally expected Euclidean topology on $\M$, but its de Groot dual. This is surprising, because
the received topology seems to be more closely related to the way, how the philosophy of physics
traditionally understands the infinity in a~context of expected finiteness of the physical
quantities. As it was remarked by de Groot in \cite{Gro} (and also by J. M. Aarts in oral
communication with de Groot), from the philosophical point of view, the co-compact topology is
naturally related to the concept of potential infinity -- in a contrast to the notion of actual
infinity, which is mostly used in the traditional mathematical approach. To illustrate the
difference, consider a countably infinite sequence $x_1, x_2, \dots$ of points lying on a straight
line in space or spacetime, with the constant distance between $x_i$ and its successor $x_{i+1}$.
In the usual, Euclidean topology, the sequence is divergent and it approaches to an improper point
at infinity. To make it convergent, one need to embed the space into its compactification (for
instance, the Alexandroff one-point compactification is a suitable one). The points completed by
the compactification then appear at the infinite distance from any other point of the space. On the
other hand, the co-compact topology, which locally coincides with the usual topology, is already
compact and superconnected, so the sequence $x_1, x_2, \dots$ is residually in each neighborhood of
every point. Since the co-compact topology locally coincides with the Euclidean topology, in most
cases it performs the same job, but in a ``more elegant" way -- with less open sets. Both
topologies are closely related to each other via the de Groot duality as we described in
Section~\ref{prerequisites}.

We may close the paper by returning to the question, that we stated at the beginning. The result of
our virtual experiment certainly is not a rigorous proof of the conjecture that the constructed
causal topology will fit with the reality also in more complex and more complicated physical
situations. But, at least, it confirms that notion of causal site of J. D. Christensen and L. Crane
is designed correctly. And it gives a strong reason for the believe, that the causal structure is
the primary structure of the spacetime, which also carries its topological information.

\bigskip

\bibliographystyle{amsplain}

\end{document}